\documentclass[letterpaper, 10pt, conference]{ieeeconf}  

\IEEEoverridecommandlockouts
\usepackage{graphicx,amsmath,amssymb}
\usepackage{cite}
\usepackage{comment,xcolor}
\usepackage{kotex}
\usepackage{algorithm}
\usepackage{algpseudocode}
\let\labelindent\relax
\usepackage{enumitem,lipsum}
\usepackage{url}
\usepackage{hyperref}
\usepackage{multirow}
\usepackage{makecell}
\usepackage{siunitx}
\usepackage{arydshln}
\usepackage{multirow}
\usepackage{nicematrix}

\usepackage{tikz}  
\usepackage{grffile}
\usepackage{pgfplots}
\usetikzlibrary{shapes,arrows}
\pgfplotsset{compat=newest}
\usetikzlibrary{plotmarks}
\usetikzlibrary{arrows.meta}
\usepgfplotslibrary{patchplots}
\usetikzlibrary{decorations.pathreplacing}

\usepackage{amsthm}
\usepackage{array}

\newtheorem{theorem}{Theorem}

\newtheorem{proposition}{Proposition}

\newtheorem{remark}{Remark}

\newtheorem{lemma}{Lemma}

\newcommand{\R}{\ensuremath{{\mathbb R}}}
\newcommand{\Z}{\ensuremath{{\mathbb Z}}}
\newcommand{\N}{\ensuremath{{\mathbb N}}}
\newcommand{\Q}{\ensuremath{{\mathbb Q}}}

\newcommand{\Enc}{\mathsf{Enc}}

\newcommand{\modp}{~\mathrm{mod}~}

\newcommand{\nult}{\mathrm{nullity}}
\newcommand{\Span}{\mathrm{span}}

\newcommand{\enc}{\ensuremath{\mathsf{Enc}}}
\newcommand{\dec}{\ensuremath{\mathsf{Dec}}}

\newcommand{\pack}{\ensuremath{\mathsf{Pack}}}

\newcommand{\upack}{\ensuremath{\mathsf{Unpack}}}

\newcommand{\Tr}{\ensuremath{\mathsf{Tr}}}

\newcommand{\Slot}{\ensuremath{\mathsf{Slot}}}

\newcommand{\scas}{\ensuremath{\mathsf{s}}} 
\newcommand{\scaL}{\ensuremath{\mathsf{L}}} 
\newcommand{\ini}{\ensuremath{\mathsf{ini}}}
\newcommand{\mult}{\ensuremath{\mathsf{mult}}}

\newcommand{\diag}{\ensuremath{\mathrm{diag}}}

\newcommand{\ak}{\ensuremath{\mathrm{ak}}}      
\newcommand{\pmes}{\ensuremath{\mathrm{m}}}    
\newcommand{\pa}{\ensuremath{\mathrm{a}}}
\newcommand{\pb}{\ensuremath{\mathrm{b}}}

\newcommand{\py}{\ensuremath{\mathrm{y}}}
\newcommand{\pz}{\ensuremath{\mathrm{z}}}
\newcommand{\pu}{\ensuremath{\mathrm{u}}}

\newcommand{\pH}{\ensuremath{{\mathrm{H}}}}
\newcommand{\pg}{\ensuremath{{\mathrm{G}}}}
\newcommand{\pr}{\ensuremath{{\mathrm{r}}}}

\newcommand{\ciph}{\ensuremath{\mathbf{c}}}

\newcommand{\ciphy}{\ensuremath{\mathbf{y}}}
\newcommand{\ciphz}{\ensuremath{\mathbf{z}}}
\newcommand{\ciphu}{\ensuremath{\mathbf{u}}}
\newcommand{\ciphf}{\ensuremath{\mathbf{F}}}
\newcommand{\ciphh}{\ensuremath{\mathbf{H}}}
\newcommand{\ciphg}{\ensuremath{\mathbf{G}}}

\newcommand{\tF}{\ensuremath{\bar{F}}}

\newcommand{\tG}{\ensuremath{\bar{G}}}
\newcommand{\tH}{\ensuremath{\bar{H}}}

\title{ \LARGE \bf
Taking Advantage of Rational Canonical Form for Faster Ring-LWE based Encrypted Controller with Recursive Multiplication
}
\author{Donghyeon Song$^1$, Yeongjun Jang$^1$, Joowon Lee$^1$, and Junsoo Kim$^2$
\thanks{This work was supported by the National Research Foundation of Korea (NRF) grant funded by the Korea government (MSIT) (No. RS-2024-00353032).}%
\thanks{$^{1}$ASRI, Department of Electrical and Computer Engineering,
Seoul National University, Seoul 08826, Korea.}
\thanks{$^{2}$Department of Electrical and Information Engineering,
Seoul National University of Science and Technology, Seoul 01811, Korea.}
\thanks{Corresponding author: Donghyeon Song {\small(dhsong@cdsl.kr)}}
}

\begin{document}

\maketitle
\thispagestyle{empty}
\pagestyle{empty}

\begin{abstract}
    This paper aims to provide an efficient implementation of encrypted linear dynamic controllers that perform recursive multiplications on a Ring-Learning With Errors (Ring-LWE) based cryptosystem.
    By adopting a system-theoretical approach, we significantly reduce both time and space complexities, particularly the number of homomorphic operations required for recursive multiplications.
    Rather than encrypting the entire state matrix of a given controller, the state matrix is transformed into its rational canonical form, whose sparse and circulant structure enables that encryption and computation are required only on its nontrivial columns.
    Furthermore, we propose a novel method to ``pack'' each of the input and the output matrices into a single polynomial, thereby reducing the number of homomorphic operations.
    Simulation results demonstrate that the proposed design enables a remarkably fast implementation of encrypted controllers.
\end{abstract}
\section{Introduction}

The notion of \emph{encrypted control} \cite{kogiso,csm,ARC, ARC23}, which has been devised to enhance the security of networked control systems,
suggests controllers that operate over encrypted data.
This means that
once control signals and parameters are encrypted, they undergo computations without being decrypted throughout the networked side, even at the controller.
Since homomorphic encryption (HE) enables the evaluation of arithmetic operations over encrypted data, it plays a key role in implementing these \emph{encrypted controllers}.

However, adapting HE to control systems has faced some challenges, especially in dealing with dynamic control operations that require recursive multiplications on encrypted data.
This is because the number of recursive multiplications is limited in most HE schemes.
The bootstrapping technique~\cite{Gentry2009Fully} can be used to extend this limit, but it is yet considered computationally expensive for real-time control systems. 
It has been shown that a linear dynamic controller needs to have an integer-valued state matrix \cite{Cheon2018Needfor}, to which the state is recursively multiplied, in order to be encrypted without bootstrapping.
Subsequently, in \cite{KimTAC2023}, an encrypted dynamic controller that
performs an unlimited number of recursive multiplications over encrypted data
has been proposed over a Learning With Errors (LWE) based cryptosystem \cite{lwe}; however, this approach still remains computationally demanding.

Recently, in \cite{Jang2024RLWE}, an encrypted controller which is also able to perform recursive multiplications for an infinite time horizon has been proposed over a Ring-Learning With Errors (Ring-LWE) based cryptosystem \cite{lyubashevsky2013ideal}, which is a faster algebraic variant of LWE based scheme. As Ring-LWE is based on polynomials, the implementation of \cite{Jang2024RLWE} utilizes a \emph{packing} method, which refers to encoding a vector of messages into a single polynomial so that the messages can be computed at once. In particular, while \cite{KimTAC2023} encrypts each component of the state matrix, \cite{Jang2024RLWE} packs and encrypts each column of the state matrix.

In this paper, we propose
an encrypted controller design over Ring-LWE based cryptosystem with recursive multiplications that achieves improved efficiency over the prior work \cite{Jang2024RLWE}, in terms of time and space complexities.
Specifically, both the computation load and memory usage of the proposed encrypted controller are $O(\log n)$ in the best case, where $n\in\N$ is the order of the controller, whereas those of \cite{Jang2024RLWE} are $O(n)$ (see Table~\ref{tab:comparison} and Remark~\ref{rem:best} for details).

To this end, we take a system-theoretical approach,
transforming
the state matrix of a given linear dynamic controller
into the \emph{rational canonical form} (also known as the Frobenius normal form).
We exploit the rational canonical form by decomposing it into a sparse matrix and a circulant matrix,
where the former consists of a few nontrivial columns.
We encrypt only these nontrivial columns and handle the circulant part utilizing the polynomial ring structure.
This significantly reduces the number of operations over encrypted data required for recursive multiplications.
In addition, such transformation is always possible,
and it is demonstrated that an integer state matrix remains to be an integer matrix after this transformation.

Furthermore, regarding the input and the output matrices, we propose customized methods to pack each of them into a single polynomial, so that fewer operations are involved for the matrix-vector multiplications compared to the prior work \cite{Jang2024RLWE}. Overall, the proposed design enables a remarkably fast implementation of encrypted controllers, as can be seen from the simulation results summarized in Table~\ref{tab:simul_time}.

The rest of this paper is organized as follows.
Section~\ref{sec:Preliminary} formulates the problem of interest and introduces the polynomial ring.
Section~\ref{sec:main}
proposes
a reformulation of a linear dynamic controller over the polynomial ring, exploiting the rational canonical form.
In Section~\ref{sec:enc}, the reformulated controller is implemented over a Ring-LWE based cryptosystem in an efficient way.
The efficiency of the proposed controller is validated through numerical simulations in Section \ref{sec:simul}.

\textit{Notation:}
The sets of integers, positive integers, rational numbers, and real numbers are denoted by $\Z$, $\N$, $\Q$, and $\R$, respectively. 
For $q\in\N$, we define $\Z_q := \Z \cap [-q/2, q/2)$.
For $a\in\Z$ and $q\in\N$, we define the modulo operation by $a \modp q := a - \lfloor (a+q/2)/q\rfloor q$.
The modulo operation is defined element-wisely for integer vectors and coefficient-wisely for polynomials.
For a polynomial $\pa(X) = \sum_{i=0}^{N-1}\pa_iX^{i}$,
let $\|\pa(X) \|:= \max_{0\le i < N}\{  | \pa_i| \}$. For a vector, $||\cdot||$ denotes the infinity norm. 
For polynomials $\pa$ and $\pb$, let $\pa\vert\pb$ refer that $\pb$ is a multiple of $\pa$.
Let $\mathbf{0}$ and $I$ denote the zero matrix and the identity matrix of appropriate sizes, respectively.
Additionally, $I_n$ denotes the $n \times n$ identity matrix for $n \in \N$.
Finally, the block diagonal matrix whose main-diagonal blocks are $C_0,\,C_1,\,\ldots,\,C_{n-1}$ is written by $\diag(C_0, C_1, \ldots, C_{n-1})$.

\section{Problem Formulation and Preliminaries} \label{sec:Preliminary}

\subsection{Problem formulation}

Consider a discrete-time linear dynamic controller
\begin{equation}\begin{split} \label{eq:ctrl}
    x(t+1) &= Fx(t) + Gy(t),  \\
    u(t) &= Hx(t), ~~ x(0) = x^\ini,
\end{split}    
\end{equation}
where $x(t)\in\R^n$ is the state with initial value $x^\ini\in\R^n$, $y(t)\in\R^p$ is the input, and $u(t)\in\R^m$ is the output.
The objective is to implement \eqref{eq:ctrl} over a Ring-LWE based cryptosystem \cite{lyubashevsky2013ideal}, in a way that i) the resulting encrypted controller is able to perform infinitely many recursive multiplications, and
ii) it achieves increased computational and memory efficiency over the prior work \cite{Jang2024RLWE}.

It is well known that the state matrix of a linear dynamic system needs to be an integer matrix in order to be realized over encrypted data \cite{Cheon2018Needfor}.
Thus, we assume that $F \in \Z^{n\times n}$.
Otherwise,
one can convert the controller to have an integer state matrix through existing methods \cite{KimjShim21,Tava22,Lee2025Integer}, or directly design a controller having an integer state matrix as in \cite{Lee2025Integer}.

For ease of analysis,
we assume that $y(t)\in\Q^p$ for all $t\ge0$, $G \in \Q^{n \times p}$, $H \in \Q^{m \times n}$, and $x^\ini \in \Q^{n}$, as $\Q$ is dense in $\R$.
As a result, $x(t) \in \Q^{n}$ and $u(t) \in \Q^{m}$ for all $t\ge0$.
In addition, let $n$ be a power-of-two.
If this is not the case, one can
increase the order of \eqref{eq:ctrl} to be the least power-of-two greater than $n$,
while the state matrix remains to be an integer matrix and the input-output relation of \eqref{eq:ctrl} is preserved.
We refer to Remark~\ref{ref:n-powerof2} for a specific method.

\subsection{Preliminaries on polynomial ring and packing}\label{subsec:RingPack}

The Ring-LWE based cryptosystem \cite{lyubashevsky2013ideal} is based on a polynomial ring, that is, 
it encrypts polynomial messages.
Thus, it is essential to encode the signals and matrices of \eqref{eq:ctrl} into polynomials, which is referred to as \emph{packing}.
In this regard, we introduce the structure of a polynomial ring
and review the packing method utilized in the prior work \cite{Jang2024RLWE}. 

Let $q \in \N$ be a prime and $N \in \N$ be a power-of-two.
{Every element of the polynomial ring $R_q := \Z_q[X]/\langle X^N +1 \rangle$ is}
represented by a polynomial with coefficients in $\Z_q$ and degree less than $N$.
{The ring $R_q$ is closed under the addition and the multiplication defined as follows; for $\pa=\sum_{i=0}^{N-1}\pa_iX^i\in R_q$ and $\pb=\sum_{i=0}^{N-1}\pb_iX^i\in R_q$,} let
\begin{equation}  \label{eq:arithPol}
    \pa + \pb \!:= \!\textstyle\sum_{i=0}^{N-1} (\pa_i + \pb_i) X^i \modp q,\,\,\,\,
    \pa \cdot \pb \!:=\! \textstyle\sum_{i=0}^{N-1} \mathrm{c}_i X^i
\end{equation}
where $\mathrm{c}_i := \sum_{j=0}^i \pa_{i-j}\pb_j - \sum_{j = i +1 }^N \pa_{N+i-j}\pb_j \modp q$.
{The multiplication over $R_q$ is in fact a typical polynomial multiplication, {mapping the coefficients into $\Z_q$} and regarding $X^N$ as $X^{N}=-1$.
This can also be interpreted as the product between a \emph{skew-circulant matrix} and a vector, that is,}
\begin{equation} \label{eq:skew-circ}
     \!
    \begin{bmatrix}
        \mathrm{c}_0 \\ \mathrm{c}_1 \\ \vdots \\ \mathrm{c}_{N-1}
    \end{bmatrix} \!\!=\!\! \begin{bmatrix}
        \pa_0 &\!\!\! -\pa_{N-1} &\!\! \cdots &\!\!\! -\pa_1 \\
        \pa_1 &\!\!\! \pa_0 &\!\! \cdots &\!\!\! -\pa_2 \\
        \vdots &\!\!\! \vdots &\!\! \ddots &\!\!\! \vdots \\
        \pa_{N-1} &\!\!\! \pa_{N-2} &\!\! \cdots &\!\!\! \pa_0
    \end{bmatrix} \!\!\begin{bmatrix}
        \pb_0 \\ \pb_1 \\ \vdots \\ \pb_{N-1}
    \end{bmatrix} \!\modp q.
\end{equation}

Now we review the packing method used in \cite{Jang2024RLWE}.
Since $N$ is typically chosen as a large number to ensure security \cite{HEstandard},
let $N/n\in \N$, which is also a power-of-two.
The \emph{coefficient packing} operation $\pack:\Z_q^n \to R_q$ is defined by
\begin{equation}
    \pack(a) := a_0 + a_1X^{\frac{N}{n}} + \cdots + a_{n-1}X^{(n-1)\frac{N}{n}} \in R_q,
\end{equation}
for a vector $a=[a_0,\,\ldots,\,a_{n-1}]^\top\in\Z_q^n$.
This packing embeds elements of a vector to a polynomial as
$N/n$-equidistant coefficients, which are often referred to as \emph{the packing slots}.
For notational simplicity, we often
write $\pack(a)$ instead of 
$\pack(a\modp q)$ for $a \in \Z^n$.

Additionally, given a power-of-two $\alpha \in \N$,
we define a mapping $\Slot_\alpha:R_q \to R_q$ for $\pa=\sum_{i=0}^{N-1}\pa_iX^i\in R_q$ by
\begin{equation}\label{eq:slotunpack}
\Slot_\alpha(\pa) :=\textstyle \sum_{i=0}^{\alpha-1}\pa_{i\frac{N}{\alpha}}X^{i\frac{N}{\alpha}},
\end{equation}
where all other coefficients are eliminated except the $N/\alpha$-equidistant ones.
For example, $\Slot_n(\cdot)$ returns a polynomial with only the packing slots, and $\Slot_1(\cdot)$ extracts the constant term of a polynomial.
Thus, it is easily derived that 
\begin{equation}
    \Slot_1(X^{-i}\cdot \pa ) = \Slot_1(-X^{N-i} \cdot \pa)  = \pa_i,
\end{equation} for $i = 0, 1, \ldots, N-1$, under the relation $X^N = -1$.

\subsection{Prior work}

We briefly review how \cite{Jang2024RLWE} utilized the packing method to perform the matrix-vector multiplications of \eqref{eq:ctrl} over $R_q$, focusing on the recursive multiplication $Fx(t)$.
For a vector $z=[z_0,z_1,\cdots,z_{n-1}]^\top\in\Z_q^{n}$ and $F \in \Z^{n\times n}$, consider the following computation:
\begin{equation}
z^+:=Fz\modp q = \textstyle\sum_{i=0}^{n-1} F_i z_i \modp q\in\Z_q^n,
\end{equation}
where $F_i$ is the $(i+1)$-th column of $F$ for $i=0,\ldots, n-1$.
One can compute $z^+$ over $R_q$ as
\begin{align}\label{eq:packMult}
    \pack(z^+)= \textstyle\sum_{i=0}^{n-1} \pack(F_i )\cdot z_i,  
\end{align}
whose right-hand-side involves $n$-multiplications of polynomials.
This technique has also been employed to compute $Gy(t)$ and $Hx(t)$, packing each and every column of the matrices into a single polynomial.

The polynomial $\pack(z^+)$ in \eqref{eq:packMult} can again be utilized for recursive multiplications, that is, to compute $\pack(Fz^+)$ as
\begin{equation}\label{eq:packRecurs} \begin{split}
    \pack(Fz^+ ) &= \textstyle\sum_{i=0}^{n-1} \pack(F_i )\cdot z^+_i,
\end{split}
\end{equation}
where $z^+_i \!:= \!\Slot_1(X^{-i\frac{N}{n}}\!\cdot\!\pack(z^+))$ for $i=0,\ldots,n-1$,
i.e., the $(i+1)$-th packing slot of $\pack(z^+)$.
The operation $\Slot_1(\cdot)$ plays the role of \emph{unpacking} $\pack(z^+)$.
Therefore, the computation of \eqref{eq:packRecurs} not only involves $n$-multiplications, but also requires the \emph{slot} operation.

\section{Rational Canonical Form and Controller Reformulation}\label{sec:main}
We propose a method to reformulate the controller \eqref{eq:ctrl} to operate over $R_q$ with fewer operations by taking advantage of the rational canonical form, which will be further elaborated in Section~\ref{subsec:RCF}.
Moreover, we propose novel packing methods that are specifically designed to encode each of $G$, $H$, and $y(t)$ into a single polynomial,
which allow the matrix-vector multiplications $Gy(t)$ and $Hx(t)$ of \eqref{eq:ctrl} to be evaluated each by a single polynomial multiplication.

As a motivating example, suppose that the state matrix $F$ is given as a \emph{companion matrix} as below, which can be
decomposed into
a sparse matrix and a skew-circulant matrix as follows:
\begin{equation}\label{eq:compan}
    F ={
    \begin{bmatrix}
    \begin{matrix}
        \!\vert \\ 
        F_0
    \end{matrix} &\!\!\!\!\!\! I_{n-1} \\
    \begin{matrix}
        \!\vert
    \end{matrix} &\!\!\!\!\!\!\! \mathbf{0}
    \end{bmatrix}
    } 
    =  
    \underbrace{\begin{bmatrix}
    \begin{matrix}
        \!\vert \\ F_0' \\ \!\vert
    \end{matrix}
    &\!\!
    \mathbf{0}~~~
\end{bmatrix}}_{\text{(sparse)}} + \underbrace{\begin{bmatrix}
        \mathbf{0}&\!\!\!\! I_{n-1} \\
    -1 &\!\!\!\! \mathbf{0}
\end{bmatrix}}_{=:S},
\end{equation}
where $F_0':=F_0 +[0,\ldots,0,1]^\top\in\Z^{n}$.
Then, it follows from the definition of $S$ in \eqref{eq:compan} and \eqref{eq:skew-circ} that given $\pack(z^+)$,
\begin{equation}
     \pack(Sz^+)= X^{-\frac{N}{n}}\cdot \pack(z^+).
\end{equation}
Thus, the left-hand-side of \eqref{eq:packRecurs} can also be computed as
\begin{equation}\label{eq:packMultRe}  
   \pack(Fz^+) = \pack(F_0')\cdot z_0^+ 
    +X^{-\frac{N}{n}}\cdot \pack(z^+),
\end{equation}
which only involves two polynomial multiplications and
the constant term $z_0^+$ of $\pack(z^+)$, in contrast to \eqref{eq:packRecurs}.

Yet, the above observation cannot be directly applied to \eqref{eq:ctrl}
because in general the state matrix $F$ is not given as a companion matrix or similar to it.
In fact,
a matrix is similar to a companion matrix if and only if its characteristic and minimal polynomials coincide \cite[Section 7.2]{hoffman1971algebra}.
Moreover, it should be guaranteed that the state matrix $F$ remains to be an integer matrix after the similarity transformation.

Nonetheless, interestingly, we show
that it is always possible to convert $F$ into an integer-valued \emph{rational canonical form},
a block diagonal matrix whose main-diagonal blocks are companion matrices.
In the sequel, we extend the idea of the motivating example to the general setup.

\subsection{Rational canonical form} \label{subsec:RCF}

The rational canonical form is defined by the following proposition, which also ensures its existence and uniqueness.

\begin{proposition}\cite[Theorem 12.16]{dummit2004abstract} \label{prop:RCF}
    Given $F \in \Q^{n\times n}$, there exists a nonsingular matrix $T \in \Q^{n\times n}$ such that
    \begin{equation}\label{eq:FRCF}
        \bar{F}:=TFT^{-1}=\diag(C_0, C_1, \ldots, C_{\kappa-1}),
    \end{equation}
    with some $\kappa \in\N$, where $C_i$ is a companion matrix for $i = 0, 1,\ldots, \kappa-1$ such that {$\det(sI - C_i)|\det(sI - C_{i+1})$ for $i = 0, 1, \ldots, \kappa-2$}. Moreover, the matrix $\bar{F}$ is uniquely determined by $F$.
\end{proposition}

The matrix $\bar{F}$ in Proposition~\ref{prop:RCF} is called the rational canonical form of $F$ over the field $\Q$. We present a constructive algorithm to find a transformation $T$ in the Appendix. Moreover, the following proposition guarantees that the rational canonical form of $F$ is an integer matrix if $F$ consists of integers.

\begin{proposition} \label{prop:RCF_integer}
    The rational canonical form of an integer matrix $F\in \Z^{n\times n}$ is also an integer matrix.  
\end{proposition}
\begin{proof}
    Consider the rational canonical form \eqref{eq:FRCF} of $F$.
    Then, its characteristic polynomial
    $\chi_F(s) := \prod_{i=0}^{\kappa-1}\det(sI-C_i)$ is an integer polynomial.
    By Gauss' lemma \cite[Theorem 2.1]{marcus1977number}, every monic divisor of $\chi_F(s)$ over the field $\Q$ is also an integer polynomial, which implies that $\det(sI-C_i)$ is an integer polynomial for $i=0,\,\ldots,\,\kappa-1$.
    Since the first column of a companion matrix consists of the coefficients of its characteristic polynomial, the proof is concluded.
\end{proof}

Similarly to \eqref{eq:compan}, we can represent the rational canonical form $\bar{F}$ as a sum of a sparse matrix and $S$;
for example, when $n=4$, $\kappa=2$, and $\det(sI-C_0) = \det(sI-C_1) = s^2-s-2$,
\NiceMatrixOptions
{
    custom-line ={command= H, tikz= dashed, width= 1mm}, 
    custom-line = {letter= I, tikz= dashed, width= 1mm}, 
}
\NiceMatrixOptions{cell-space-limits=1pt}
\begin{align}\label{eq:F_example} 
    \bar{F}=\begin{bNiceArray}{cc I cc }[margin=2pt]
     1 & 1 & 0 & 0 \\
     2 & 0 & 0 & 0 \\\H
     0 & 0 & 1 & 1 \\
     0 & 0 & 2 & 0
\end{bNiceArray} = 
\begin{bNiceArray}{c I c I c I c }[margin=2pt] 
     1 & 0 & 0 & 0 \\
     2 & 0 & -1 & 0 \\
     0 & 0 & 1 & 0 \\
     1 & 0 & 2 & 0
\end{bNiceArray} + S.
\end{align}

\begin{remark}
    The number of nontrivial columns in the sparse matrix $\bar{F} - S$ is $\kappa$, the number of the companion blocks in the rational canonical form $\bar{F}$.
    By Proposition~\ref{prop:kappa} in the Appendix,
    $
    \kappa =  \max_{\lambda \in \mathbb{C}} \mathrm{nullity}(F-\lambda I)
    $.
\end{remark}

\subsection{Reformulation of controller over \texorpdfstring{$R_q$}{Rq}}

We begin by converting the controller \eqref{eq:ctrl} to operate over $\Z_q$, so that its signals and matrices can be packed into polynomials in $R_q$.
With the messages
$\{H,G,y(t),x^\ini\}$
and the transformation $T$ in \eqref{eq:FRCF} 
consisting of rationals, let
\begin{equation} \label{eq:quantize_param}
\begin{split}
        \tG &:= \frac{TG}{\scas_1} \in \Z^{n\times p},  ~~~~~\;\;\tH := \frac{HT^{-1}}{\scas_2} \in \Z^{m\times n},  \\
    z^\ini &:= \frac{Tx^\ini}{\scaL\scas_1} \in \Z^n,  ~~~~~\;\bar{y}(t):=\frac{y(t)}{\scaL} \in \Z^p,
\end{split}
\end{equation}
with some scale factors $1/\scas_1\in\N$, $1/\scas_2\in\N$, and $1/\scaL\in\N$.
Then, the controller \eqref{eq:ctrl} can be implemented over $\Z_q$, as
\begin{equation}\begin{split} \label{eq:ctrl_TF}
    z(t+1) &= \tF z(t) + \tG \bar{y}(t) \modp q, \\
    \bar{u}(t) &= \tH z(t) \modp q, ~~ z(0) = z^\ini \modp q.
\end{split}    
\end{equation}

\begin{remark} \label{ref:n-powerof2}
    If $n$ is not a power-of-two,
    then one can
    {increase the order of \eqref{eq:ctrl_TF} to $\bar{n}:=2^{\lceil \log_2 n\rceil}$,
    by replacing
    $C_{\kappa-1}$ of $\tF$ with another companion matrix $C_{\kappa-1}^\prime$ such that}
    $\det(sI-C_{\kappa-1}^\prime)=s^{\bar{n}-n}\det(sI-C_{\kappa-1})$.
    The matrices $\tG$ and $\tH$ can be expanded
    as $[\tG^\top,\, \mathbf{0}]^\top$ and $[\tH,\,\mathbf{0}]$, respectively.
\end{remark}

Now we reformulate the state matrix-vector multiplication $\bar{F}z(t)$ into polynomial multiplications over $R_q$. 
Note that the matrix $\tF-S$ is sparse, having only $\kappa$-nontrivial columns, which we denote by $\bar{F}_{0}',\bar{F}_{1}',\ldots,\bar{F}_{{\kappa-1}}' \in \Z^n$. We also define $r_i$ such that $\bar{F}_i'$ is the $(r_i + 1)$-th column of $\tF - S$ for $i=0,\ldots,\kappa-1$.
For example in \eqref{eq:F_example}, $r_0 = 0$, $r_1 = 2$, $\bar{F}_0'=\left[1,\,2,\,0,\,1\right]^\top$,  and $\bar{F}_1'=\left[0,\,-1,\,1,\,2\right]^\top$.
Then, $\bar{F}z(t)$ can be reformulated over $R_q$, analogously to \eqref{eq:packMultRe}, as
\begin{align}\label{eq:pack_Fz} 
    \tilde{\mathrm{F}}_i &:= \pack(\bar{F}^\prime_i),\quad i=0,\,\ldots,\,\kappa-1, \\
    \pack(\tF z(t)) &=\textstyle \sum_{i=0}^{\kappa -1} \tilde{\mathrm{F}}_i  \cdot z_{r_i}(t) \!+\! X^{-\frac{N}{n}} \cdot \pack(z(t)), \nonumber
\end{align}
where $z_{r_i}(t) := \Slot_1(X^{-r_i\frac{N}{n}}\cdot \pack(z(t)))\in\Z_q$.
Compared to \eqref{eq:packRecurs}, it should be noted that \eqref{eq:pack_Fz} only involves $(\kappa+1)$-polynomial multiplications and only unpacks $\kappa$-packing slots $\{z_{r_i}(t) \}_{i=0}^{\kappa-1}$ of $\pack(z(t))$.

Next, we present novel packing methods that enable each of the non-recursive parts $\bar{G}\bar{y}(t)$ and $\bar{H}z(t)$ in \eqref{eq:ctrl_TF} to be computed by a single polynomial multiplication.
Let the input $\bar{y}(t) =: [\bar{y}_0(t) , \bar{y}_1(t), \ldots, \bar{y}_{p-1}(t)]^\top $ be packed as 
\begin{equation}\label{eq:pack_y}
    \!\!\!\tilde{\py}(t)
    \!:=\! \bar{y}_0(t) + \bar{y}_1(t)X + \cdots + \bar{y}_{p-1}(t)X^{p-1}\modp q \in\! R_q.\!
\end{equation}
Also, let each row of $\bar{G}=:(\bar{G}_{i,j}) \in \Z^{n\times p}$ be packed as
\begin{align}
    \tilde{\pg}_i &:= \bar{G}_{i,0}+\bar{G}_{i,1}X^{-1} \! +\!\cdots\! +\! \bar{G}_{i,p-1}X^{-(p-1)} \modp q \in R_q,
\end{align}
for $i=0,\ldots,n-1$.
Then, it follows from \eqref{eq:arithPol} that 
\begin{equation}\label{eq:gRowMult}
    \tilde{\pg}_i\cdot\tilde{\py}(t) = \textstyle
    \sum_{j=0}^{p-1} \bar{G}_{i,j}\bar{y}_j(t) + \textstyle\sum_{j=-(p-1), \, j\ne 0}^{p-1} \pr_{i,j} X^{j}
\end{equation}
for some $\pr_{i,j}\in\Z_q$,
where the constant term $\sum_{j=0}^{p-1} \bar{G}_{i,j}\bar{y}_j(t) $ is the $(i+1)$-th element of $\bar{G}\bar{y}(t)$.
Based on this observation, we pack the entire matrix $\bar{G}$ into a single polynomial, as
\begin{equation}\label{eq:pg}
    \tilde{\pg}  :=\textstyle\sum_{i=0}^{n-1}\tilde{\pg}_i\cdot X^{i\frac{N}{n}}.
\end{equation}
This enables $\bar{G}\bar{y}(t)$ to be computed as
\begin{align}\label{eq:trick_gy}
    \Slot_n(\tilde{\pg} \cdot  \tilde{\py}(t))& = \Slot_n(\textstyle\sum_{i=0}^{n-1} (\tilde{\pg}_i\cdot \tilde{\py}(t))\cdot X^{i\frac{N}{n}} )\nonumber\\
    &=\pack(\bar{G}\bar{y}(t)),
\end{align}
assuming
that $np \le N$; this assumption prevents the coefficients $\pr_{i,j}$ in \eqref{eq:gRowMult} from affecting the packing slots. Note that $N$ is typically chosen as a large number to ensure security.

For the matrix $\tH=:(\bar{H}_{i,j})\in \Z^{m \times n}$, let each row be packed as 
\begin{equation}
    \tilde{\pH}_i \!:=\! \bar{H}_{i,0} + \bar{H}_{i,1}X^{-\frac{N}{n}} \!+ \cdots + \bar{H}_{i,n-1}X^{-(n-1)\frac{N}{n}}\modp q \in\! R_q,
\end{equation}
for $i=0,\ldots,m-1$.
Then, it is immediate from \eqref{eq:arithPol} and the definition of $\pack(\cdot)$ that for any $z=[z_0,\ldots,z_{n-1}]^\top\in\Z_q^n$,
\begin{equation}\label{eq:HzPack}
    \tilde{\pH}_i\cdot \pack(z) =  \textstyle \sum_{j=0}^{n-1} \bar{H}_{i,j}z_j + \sum_{j=1}^{n-1} \pr_j X^{j\frac{N}{n}},
\end{equation}
for some $\pr_j\in \Z_q$, where the constant term $\sum_{j=0}^{n-1} \bar{H}_{i,j}z_j$ is the $(i+1)$-th element of $\bar{H}z$.
Using $\{\tilde{\pH}_i\}_{i=0}^{m-1}$, we define the packed polynomial of $\bar{H}$ by
\begin{equation} \label{eq:pack_h}
    \tilde{\mathrm{H}} := \textstyle\sum_{i=0}^{m-1}\tilde{\pH}_i \cdot X^{i\frac{N}{n\tau}} \in R_q,
\end{equation}
where $\tau = 2^{\lceil \log_2m\rceil}$ is the least power-of-two not less than $m$.
We make a mild assumption that $\tau \le N/n$, which ensures that the elements of $\bar{H}$ do not overlap and reside in the $N/(n\tau)$-equidistant coefficients of $\tilde{\pH}$.
Then, with an \textit{unpacking operation} $\upack:R_q \to \Z_q^m$ defined by
\begin{equation} \label{eq:unpack_u}\begin{split}
\upack(\pa) := [\pa_0, \pa_\frac{N}{n\tau}, \ldots, \pa_{(m-1)\frac{N}{n\tau}}]^\top,
\end{split} \end{equation}
for a polynomial $\pa(X) = \sum_{i=0}^{N-1}\pa_iX^i$, the matrix-vector multiplication $\bar{H}z$ can be computed over $R_q$ by means of
\begin{align}\label{eq:trick_hx}
    &\upack(\tilde{\pH}\cdot\pack(z)) \\ &=\upack(\textstyle\sum_{i=0}^{m-1}(\tilde{\pH}_i\cdot \pack(z)) \cdot X^{i\frac{N}{n\tau}}) = \bar{H}z\modp q.\nonumber
\end{align}

Finally, combining \eqref{eq:pack_Fz}, \eqref{eq:trick_gy}, and \eqref{eq:trick_hx}, we reformulate the controller \eqref{eq:ctrl_TF} to operate over $R_q$, as
\begin{subequations}\label{eq:ctrl_packed}
\begin{align}
    \tilde{\pz}(t+1) \!&=\! \textstyle\sum_{i=0}^{\kappa -1} \tilde{\mathrm{F}}_i \cdot
    \tilde{\pz}_{r_i}(t)\!
    +\!X^{-\frac{N}{n}}\cdot\tilde{\pz}(t) \!+\!  \tilde{\pg}\cdot \tilde{\py}(t) ,\label{eq:ctrl_packedstateeq}\\
    \tilde{\pu}(t) \!&=\! \tilde{\mathrm{H}} \cdot \Slot_n(\tilde{\pz}(t)),~~\tilde{\pz}(0) \!= \!\pack( z^\ini).\label{eq:ctrl_packed_outputeq}
\end{align}
where $\tilde{\pz}_{r_i}(t):=\Slot_1(X^{-r_i\frac{N}{n}}\cdot \tilde{\pz}(t))\in \Z_q$. 
The controller output $u(t)\in \Q^m$ can be obtained from $\tilde{\pu}(t)$, by
\begin{align}
    u(t) = \left(\scaL\scas_1\scas_2\right)\cdot\upack (\tilde{\pu}(t)). \label{eq:u_packed}
\end{align}
\end{subequations}

The following theorem states the correctness of the proposed reformulated controller \eqref{eq:ctrl_packed}, as long as the modulus $q$ is chosen sufficiently large to prevent an overflow, which is indeed clear from the construction.

\begin{theorem} \label{thm:correctness}
    Consider the original controller \eqref{eq:ctrl} and the reformulated controller \eqref{eq:ctrl_packed}. If 
    \begin{equation} \label{eq:thm_bdd}
        \sup_{t \ge 0} \left\{\bigg\lVert \frac{Tx(t)}{\scaL\scas_1} \bigg\rVert , \bigg\lVert \frac{u(t)}{\scaL\scas_1\scas_2} \bigg\rVert \right\} < \frac{q}{2},
    \end{equation}
    then the output $u(t)$ of the reformulated controller \eqref{eq:u_packed} equals to that of the original controller \eqref{eq:ctrl}, for all $t \ge 0$.  
\end{theorem}

\begin{proof}
    The boundedness condition \eqref{eq:thm_bdd} ensures that the original controller \eqref{eq:ctrl} and the controller \eqref{eq:ctrl_TF} over $\Z_q$ are equivalent, that is $\scaL\scas_1\cdot z(t) = Tx(t)$ and $\scaL\scas_1\scas_2\cdot \bar{u}(t) = u(t)$ for all $t \ge 0$, since the modular operation in \eqref{eq:ctrl_TF} can be removed. Also, with \eqref{eq:pack_Fz}, \eqref{eq:trick_gy}, and \eqref{eq:trick_hx}, it follows that $\Slot_n(\tilde{\pz}(t)) = \pack(z(t))$ and $\bar{u}(t)= \upack(\tilde{\pu}(t))$, for all $t \ge 0$, and this concludes the proof.
\end{proof}

\begin{remark}
    The condition \eqref{eq:thm_bdd} is quite a natural one; if the controller stabilizes a plant, then the closed-loop stability implies the boundedness. What remains is to choose the modulus $q$ as a sufficiently large prime number.
\end{remark}

\section{Encrypted Controller Design} \label{sec:enc}

Finally,
we implement the {reformulated controller \eqref{eq:ctrl_packed}} over the Ring-LWE based cryptosystem.
It will be shown that the effect of \emph{error growths} due to encryptions and homomorphic operations can be interpreted as bounded perturbations applied to {\eqref{eq:ctrl_packed}}.
The efficiency of the proposed encrypted controller is discussed in terms of required computational burden and memory, compared to the prior work \cite{Jang2024RLWE}.

\subsection{Ring-LWE based cryptosystem} \label{subsec:rlwe}

We utilize the Ring-LWE based scheme of \cite{lyubashevsky2013ideal}, along with the external product \cite{chillotti2016rGSW} for multiplications.
We further adopt the special modulus technique \cite{GentHale12} for the external product, which reduces the growth of errors occurred by the external product.
To begin with, the set of parameters $\lbrace \nu, P, \sigma \rbrace$ is selected; $\nu\in\N$ is a power-of-two, the special modulus $P\in \N$ is a prime number, and $\sigma >0$ is a bound of the \textit{error polynomials}, which are injected during encryption in order to ensure security.
Let $d := \lceil \log_\nu q\rceil$.

\begin{table}
    \centering
    \vspace{5pt}
    \caption{Algorithms of Ring-LWE based cryptosystem}
    \label{tab:RLWE}
    \renewcommand{\arraystretch}{1.4}
    \begin{tabular}{l|l}
       \hline
        & Algorithm \\
        \hline
        Encryption & $\enc: R_q \to R_q^2$\\
        Decryption & $\dec: R_q^2 \to R_q$\\
        Encryption for multiplier & $\enc': R_q \to R_{qP}^{2\times 2d}$\\
        External product & $\boxdot:R_{qP}^{2\times 2d} \times R_q^2 \to R_q^2$\\
        Automorphism for plaintext & $\Psi_\theta(\cdot):R_q\to R_q$ \\
        Automorphism for ciphertext & $\Phi_\theta(\cdot,\ak_\theta):R_q^2 \to R_q^2$\\
        \hline
        \multicolumn{2}{l}{* $\theta\in\N$ is an odd number, and $\ak_\theta \in R_{qP}^{2 \times 2d}$ is the automorphism key.}
    \end{tabular}
\end{table}

The algorithms of the Ring-LWE based cryptosystem that we use are listed in Table~\ref{tab:RLWE},
where we follow the same notation in \cite{Jang2024RLWE} for consistency.
The algorithm $\enc$ encrypts \textit{plaintexts} (messages) in $R_q$ to be \textit{ciphertexts} (encrypted data) in $R_q^2$.
The algorithm $\enc^\prime$ returns a ciphertext in $R_{qP}^{2\times 2d}:= (\Z_{qP}[X]/\langle X^N+1\rangle)^{2\times 2d}$ that is used as a multiplier of the external product.
The automorphism for plaintext operates as $\Psi_\theta(\pa(X)) = \pa(X^\theta)$ where $\pa \in R_q$ and $\theta \in \N$ is an odd number.
Note that the automorphism $\Phi_\theta(\cdot, \ak_\theta)$ for ciphertexts requires the \textit{automorphism key} $\ak_\theta$.

The algorithms of Table~\ref{tab:RLWE} satisfy the following properties \cite{Jang2024RLWE}
for any $\pmes\in R_q$, $\ciph\in R_q^2$, and $\ciph'\in R_q^2$:
\begin{enumerate}[align=left,leftmargin=*]
    \item [(H1)] $\left\| \dec(\enc(\pmes)) - \pmes \right\| \le \sigma$.
    \item[(H2)] $\dec(\ciph + \ciph') =\dec(\ciph) + \dec(\ciph')$.
    \item[(H3)] $\dec(\pmes\cdot \ciph) =\pmes \cdot \dec(\ciph)$.
    \item[(H4)] $\left\| \dec(\enc'(\pmes)\boxdot \ciph) - \pmes \cdot \dec(\ciph) \right\| \le \sigma_\mult := {P^{-1}dN\sigma\nu + (N+1)/2}$.
    \item[(H5)] $\left\|\dec(\Phi_\theta(\ciph, \ak_\theta)) - \Psi_\theta(\dec(\ciph)) \right\| \le\sigma_\mult$, where $\theta\in\N$ is an odd number.
\end{enumerate}
The property (H1) indicates the \textit{correctness} of the scheme,
and (H2) is the \textit{additively homomorphic} property.
It is stated by (H3) that a plaintext can be multiplied to a ciphertext.
Lastly, (H4) and (H5) enable the multiplication and the automorphism to be carried out over encrypted data, respectively.

\begin{algorithm}[t]
\caption{Trace}
\begin{algorithmic}[1]
\renewcommand{\algorithmicrequire}{\textbf{Procedure}}
\renewcommand{\algorithmicensure}{\textbf{Input}}
\Ensure Power-of-twos $\alpha$ and $\beta$ such that $1 \le \alpha < \beta \le N$ and $\ciph \in R_q^2$
\Require $\Tr_{\beta}^{\alpha}(\ciph)$
\State Prepare $\ak_\theta$ for $\theta \!\in\! \{ 2^\delta + 1 \!:\! \log_2\alpha < \delta \le \log_2\beta, ~ \delta \in \mathbb{N}\}$
\State $\ciph' \gets \ciph$
\For{($k = \beta$; $k>\alpha$; $k = k/2$)}
    \State $\ciph' \gets (q+1)/2 \cdot \ciph'$
    \State $\ciph' \gets \ciph' + \Phi_{k+1}(\ciph', \ak_{k+1})$
\EndFor
\State \textbf{return} $\ciph' \in R_q^2$
\end{algorithmic}
\label{alg:Trace}
\end{algorithm}

Now we need a homomorphic evaluation of the slot operation $\Slot_\alpha$ in \eqref{eq:slotunpack}. 
In fact, it is not necessary to zero-out all the coefficients except for the $N/\alpha$-equidistant coefficients. 
To illustrate this point, observe that for any power-of-two $\alpha$, and $\pa \in R_q$ and $\pb\in R_q$, it holds that
\begin{align} \label{eq:slot_mult}
    \Slot_\alpha(\pa)\cdot \Slot_\alpha (\pb) = \Slot_\alpha(\Slot_\alpha(\pa )\cdot \pb),
\end{align}
which can be easily derived from \eqref{eq:skew-circ}.
The implication of \eqref{eq:slot_mult} is that the coefficients of $\pb$ which are not $N/\alpha$-equidistant do not affect the $N/\alpha$-equidistant ones of the multiplication result.
Hence, for example in \eqref{eq:ctrl_packed_outputeq}, it is enough to zero-out only the $N/(n\tau)$-equidistant coefficients of {$\tilde{\pz}(t)$} that are not $N/n$-equidistant since the coefficients of $\tilde{\mathrm{H}}$ are $N/(n\tau)$-equidistant. 
In this regard, we use the Trace algorithm $\Tr_\beta^\alpha(\cdot)$ in Algorithm~\ref{alg:Trace}, a slight generalization of \cite[Algorithm~1]{Chen2021Trace}, which homomorphically eliminates all the $N/\beta$-equidistant coefficients except the $N/\alpha$-equidistant ones in the sense of the following lemma.

\begin{lemma} \label{lem:tr}
Given power-of-twos $\alpha$ and $\beta$ such that $1 \le \alpha < \beta \le N$, it holds for any $\ciph \in R_q^2$ that
\begin{equation} \label{eq:lemtr}
\Slot_\beta(\dec(\Tr_\beta^\alpha(\ciph))) = \Slot_\alpha(\dec(\ciph)) + \Delta,
\end{equation}
for some $\Delta \in R_q$ such that $||\Delta|| < \sigma_\mult\cdot \log_2 (\beta/\alpha)$.
\end{lemma}

\begin{proof}
    The proof is a straightforward extension of the proof of \cite[Lemma 2]{Jang2024RLWE}, hence is omitted.
\end{proof}

\subsection{Encrypted controller design}

Now we propose the encrypted controller design based on the reformulated controller \eqref{eq:ctrl_packed}.
The offline procedure is as follows:
Given the controller \eqref{eq:ctrl}, a coordinate transformation matrix is found by Algorithm~\ref{alg:RCF} in the Appendix.
The transformed controller
is scaled-up as \eqref{eq:ctrl_TF}, and the matrices $\tF$, $\tG$, and $\tH$ are packed into polynomials as in \eqref{eq:pack_Fz}, \eqref{eq:pg}, and \eqref{eq:pack_h}, respectively.
We encrypt these packed control parameters as
\begin{equation}\label{eq:encparam}
    \ciphf_i:=\enc^\prime(\tilde{\mathrm{F}}_i),\quad    \ciphg:=\enc^\prime(\tilde{\pg}),\quad\ciphh:=\enc^\prime(\tilde{\mathrm{H}}),
\end{equation}
for $i=0,\,\ldots,\,\kappa-1$.
We also generate automorphism keys $\ak_\theta$ for $\theta\in\lbrace 2^\delta+1: 0 < \delta < \log_2 n\tau, \delta \in \N \rbrace$, which are required to execute Algorithm~\ref{alg:Trace}.
The initial state of \eqref{eq:ctrl_TF} is packed and then encrypted, as $\ciphz^\ini:=\enc(\pack(z^\ini))$.

From the proposed reformulation \eqref{eq:ctrl_packed}, the encrypted controller is constructed as
\begin{subequations}\label{eq:ctrl_enc}
\begin{align}
    \ciphz(t+1) &= 
    \textstyle\sum_{i=0}^{\kappa -1}\ciphf_i\boxdot \Tr_n^1(X^{-r_i\frac{N}{n}}\cdot\ciphz(t))+X^{-\frac{N}{n}}\cdot\ciphz(t)\notag \\
    &~~~~ + \ciphg \boxdot \ciphy(t),\label{eq:encFxGu}\\
\ciphu(t)&=\ciphh \boxdot \Tr_{n\tau}^n(\ciphz(t)), \qquad \ciphz(0) = \ciphz^\ini,\label{eq:encHx}
\end{align}
where
$\ciphz(t)\in R_q^2$ is the state,
$\ciphy(t)\in R_q^2$ is the input sent from the plant, and $\ciphu(t)\in R_q^2$ is the output.
Here, we do not encrypt $X^{-\frac{N}{n}}$ part of \eqref{eq:ctrl_packedstateeq} since the skew-circulant matrix $S$ essentially contains no information.
Notably, the operation $\Tr_n^1(\cdot)$ is applied to $\ciphz(t)$ to homomorphically unpack only $\kappa$-packing slots, each multiplied with the packed and encrypted nontrivial columns $\ciphf_i$, taking advantage of the rational canonical form.

At each time step, the plant output $y(t)$ is scaled, packed as in \eqref{eq:pack_y}, and then encrypted as $\ciphy(t)=\enc(\tilde{\py}(t))$.
The output $\ciphu(t)$ is sent to the actuator, where it is decrypted, unpacked, and scaled to obtain the control input $u(t)$, as
\begin{align} \label{eq:act}
    u(t) = \left(\scaL\scas_1\scas_2\right)\cdot\upack(\dec(\ciphu(t))).
\end{align} 
\end{subequations}

\begin{proposition} \label{thm:error_growth1}
Consider the encrypted controller \eqref{eq:ctrl_enc}, and define $\bar{\pz}(t) := \dec(\ciphz(t))$ for $t \ge 0$. Then,
compared to \eqref{eq:ctrl_packed},
the messages of the encrypted controller \eqref{eq:ctrl_enc} 
obey
\begin{align}
    \hspace{-3pt} \bar{\pz}(t+1) \!&=\! \textstyle\sum_{i=0}^{\kappa -1} \tilde{\mathrm{F}}_i \cdot
    \bar{\pz}_{r_i}(t)\!  \nonumber
    +\!X^{-\frac{N}{n}}\cdot\bar{\pz}(t) \!+\!  \tilde{\pg}\cdot \tilde{\py}(t) \!+\!\Delta_z(t),  \\
    u(t)\! &= \!\left(\scaL\scas_1\scas_2\right)\cdot\upack (\tilde{\mathrm{H}} \cdot \Slot_n(\bar{\pz}(t))) \!+\! \Delta_u(t), \label{eq:ctrl_pert} \\
    \bar{\pz}(0) \!&= \!\pack( z^\ini) \!+\!\Delta_\ini,  \nonumber
\end{align}
where $\bar{\pz}_{r_i}(t):=\Slot_1(X^{-r_i\frac{N}{n}}\cdot \bar{\pz}(t))\in \Z_q$ and the perturbations are bounded as
    \begin{equation} \begin{split} 
        \hspace{-2pt}||\Slot_n(\Delta_z(t))|| &\le \left(\textstyle\sum_{i=0}^{\kappa -1}||\tilde{\mathrm{F}}_i||\cdot n\log_2n + \kappa + 1  \right)\!\cdot\!\sigma_\mult \\&~~+ np||\tilde{\mathrm{G}}|| \cdot \sigma, \\
        ||\Delta_u(t)|| &\le \left(\scaL\scas_1\scas_2\right)\cdot\left( 1 + ||\tilde{\mathrm{H}}|| \cdot  nm\log_2\tau \right)\! \cdot\! \sigma_\mult, \\
        ||\Delta_\ini|| &\le \sigma.
    \end{split}
\end{equation}
\end{proposition}
\begin{proof}
    The proof is omitted, as it can be analogously shown with Lemma~\ref{lem:tr} and property \eqref{eq:slot_mult} along the proof of \cite[Lemma 3]{Jang2024RLWE}.
\end{proof}

Though the perturbation $\Delta_z(t)$ in \eqref{eq:ctrl_pert} outside the packing slots can be unbounded, it never affects the input-output relationship and is therefore out of our interest.
More importantly, although this theorem only shows the boundedness of perturbations, their effects on stable closed-loop system can be arbitrarily suppressed by adjusting the scaling factor $\scaL$. 
In particular, in \eqref{eq:quantize_param}, the factor $1/\scaL$ was chosen just large enough to convert the matrices and signals into integers. 
However, by increasing $1/\scaL$ further, the impact of perturbations can be made arbitrarily small.
For the selection of the factor $\scaL$ and the modulus $q$, we refer to \cite[Theorem 2]{Jang2024RLWE}.

\subsection{Discussion}

The improved efficiency of the proposed encrypted controller \eqref{eq:ctrl_enc} over the prior work \cite{Jang2024RLWE} is discussed in terms of computational load and memory consumption.
First, we examine the computational load of the encrypted controller by the number of external products executed at each time step, since the external product is one of the most computationally demanding operations among the homomorphic algorithms in Table~\ref{tab:RLWE} \cite{chillotti2016rGSW}.
In addition, each automorphism requires one external product, so Algorithm~\ref{alg:Trace} demands $\log_2(\beta/\alpha)$-external products.
Therefore, there are a total of $(2 + \kappa(1+\log_2n)+\lceil\log_2m\rceil)$-external products in the proposed controller \eqref{eq:ctrl_enc},
$1+\kappa(1+\log_2n)$ for the state equation \eqref{eq:encFxGu} and $1+\lceil \log_2m\rceil$ for the output equation \eqref{eq:encHx}, respectively.

Next, the memory consumption is quantified by the total number of encrypted control parameters, as the ones in \eqref{eq:encparam}, and the automorphism keys that the controller must store for its computation.
Recall that the parameters encrypted by $\Enc^\prime(\cdot)$ and the automorphism keys all belong to $R_{qP}^{2\times 2d}$.
The proposed controller \eqref{eq:ctrl_enc} possesses $(\log_2n+\lceil \log_2m\rceil)$-automorphism keys and $(\kappa+2)$-encrypted parameters of \eqref{eq:encparam}.

On the other hand, the prior work \cite{Jang2024RLWE} packs each and every column of the control parameters, which leads to computation and memory consumption proportional to both the degree of the controller
$n$ and the input dimension
$p$.
Table~\ref{tab:comparison} summarizes this comparison and highlights the effectiveness of the proposed design.

\begin{remark}\label{rem:best}
    The proposed design has a great advantage if the controller \eqref{eq:ctrl} has a single output from which it is observable, so that it can be transformed into the observable canonical form, that is, $\kappa = m = 1$.
    In this case, both the computational load and the memory consumption of the proposed design in Table~\ref{tab:comparison} reduce to $3+\log_2n$, whereas those of \cite{Jang2024RLWE} are $O(n)$.
\end{remark}

\begin{table}
    \vspace{5pt}
    \caption{Efficiency of encrypted controllers in terms of computation and memory consumption}
    \label{tab:comparison}
    \centering
    \renewcommand{\arraystretch}{1.2}
    \begin{tabular}{c||c|c}
         \hline
         & Computational load & Memory consumption\\
         \hline
        \cite{Jang2024RLWE} & $4n + p - 2$ & $2n+p+\log_2 n$\\
        \hline
        \multirow{2}{*}{Proposed} & $2 + \kappa(1+\log_2n)$ & $\kappa+2+\log_2n$\\
         & $+\lceil\log_2m\rceil$& $+\lceil\log_2m\rceil$\\
        \hline
    \end{tabular}
\end{table}

\section{Simulation Results} \label{sec:simul}

The efficiency of the proposed encrypted controller 
\eqref{eq:ctrl_enc} is demonstrated through simulation results.
A discrete-time plant
stabilized by the controller \eqref{eq:ctrl} is written as
\begin{equation}\label{eq:plant}
    \begin{split}
        x_\mathsf{p}(t+1) &= Ax_\mathsf{p}(t) + Bu(t),\\
        y(t)&= Cx_\mathsf{p}(t), ~~ x_\mathsf{p}(0) = x_\mathsf{p}^\ini,
    \end{split}
\end{equation}
where $x_\mathsf{p}(t)\in \R^{n_\mathsf{p}}$ is the state, $u(t)\in\R^m$ is the input, and $y(t)\in\R^p$ is the output\footnote{{The plant output $y(t)$ is quantized as
$y(t)\mapsto\lceil y(t)\cdot  10^5\rfloor/10^5 $,
to be a rational number.}}.
We consider the following two cases where $(n,\kappa,m,p)$ of the controller \eqref{eq:ctrl} differs:
\begin{enumerate}[wide, labelwidth=0pt, labelindent=0pt]
    \item $(n,\kappa,m,p)=(8,1,1,1)$:
    The plant \eqref{eq:plant} is obtained by discretizing the following linearized inverted pendulum model \cite{franklin2002feedback} under the sampling period $\SI{50}{\milli\s}$:
    \begin{equation*}
    \begin{aligned}
        \left(\mathsf{I}+\mathsf{m}l^2\right)\ddot{\phi}(\mathsf{t})-\mathsf{m}\mathsf{g}l\phi(\mathsf{t})&=\mathsf{m}l\ddot{\mathsf{x}}(\mathsf{t}),\\
\left(\mathsf{M}+\mathsf{m}\right)\ddot{\mathsf{x}}(\mathsf{t})+b\dot{\mathsf{x}}(\mathsf{t})-\mathsf{m}l\ddot{\phi}(\mathsf{t})&=u(\mathsf{t}),\,\,
        y(\mathsf{t})=\mathsf{x}(\mathsf{t}),
    \end{aligned}
    \end{equation*}
    where
    $\mathsf
    x(\mathsf{t})\in\R$ is the cart position, $\phi(\mathsf{t})\in\R$ is the angle of the pendulum, with $\mathsf{M}=0.5$, $\mathsf{m}=0.2$, $b=0.1$, $l=0.2$, $\mathsf{I}=0.006$, and $\mathsf{g}=9.8$.
    The controller \eqref{eq:ctrl} is designed using the method of \cite{Lee2025Integer},
    where the state matrix $F$ is a companion matrix with $\det(sI-F)=s^4\left(s^4-s^3-13s^2-4s+10\right)$, $H=[10,\,0,\,0,\,0,\,0,\,0,\,0,\,0]$, and
    \begin{align*}
        \resizebox{\columnwidth}{!}{$
        G=\left[-640.5,\,1715,\,-1489,\,389.5,\,27.23,\,-0.6047,\,-2.364,\,0.4784\right]^\top.
        $}
    \end{align*}
    The initial condition of the inverted pendulum is $\mathsf{x}(0)=\dot{\mathsf{x}}(0)=0$ and $\phi(0)=\dot{\phi}(0)=0.1$.
    The initial state of the controller \eqref{eq:ctrl} is set as $x^\ini=\mathbf{0}$.
    
    \item $(n,\kappa,m,p)=(4,2,2,2)$:
    The matrices of the plant \eqref{eq:plant} are given as
    \begin{align*}
          A &= {\footnotesize\begin{bmatrix}
        -4.9535 & -1.3701 & 2.0157 & 1.0929 \\
        5.4838 & 3.0300 & -3.8440 & -1.9888 \\
        0.9319 & 0.5722 & -1.2467 & 0.5866 \\
        -2.4378 & -0.9447 & -2.0371 & -0.6299 \\
    \end{bmatrix}},\\
        B &= {\footnotesize\begin{bmatrix}
        1.3993 & 2.0586 & 0.0968 & 0.1186 \\
        -0.0344 & -0.0405 & 0.4669 & 0.6871 \\
    \end{bmatrix}^\top}, \\
        C &= {\footnotesize\begin{bmatrix}
        -0.5224 & -0.2219 & -0.3423 & -0.1006 \\
        -0.9765 & -0.5500 & 0.8802 & 0.4234 \\
    \end{bmatrix}}.
    \end{align*}
    The controller \eqref{eq:ctrl} has the state matrix as \eqref{eq:F_example}, and the other matrices are given as
    \begin{align*}
        G &= {\footnotesize\begin{bmatrix}
        2.7 & -1.3 & -0.1 & 5.0 \\
        3.2 & -4.9 & -1.0 & -0.3 \\
    \end{bmatrix}}^\top,~
        H = {\footnotesize\begin{bmatrix}
            1 & 0 & 0 & 0 \\
            0 & 0 & 3 & 0
        \end{bmatrix}}.
    \end{align*}
    The initial values of the plant \eqref{eq:plant} and the controller \eqref{eq:ctrl} are set as $x_\mathsf{p}^\ini = [0,\,0,\,0.1,\,-0.1]^\top$ and $x^\ini=\mathbf{0}$, respectively.
\end{enumerate}

\begin{table}[t]
    \centering
    \caption{Elapsed time of encrypted controllers at each time step }\label{tab:simul_time}
    \renewcommand{\arraystretch}{1.1}
    \begin{tabular}{c|c||ccc}
        \hline
        Case & Method & Max ($\mathrm{ms}$) & Mean ($\mathrm{ms}$) & Std ($\mathrm{ms}$) \\
        \hline
        \multirow{2}{*}{1)} & \cite{Jang2024RLWE} & 31.19 & 21.81 & 3.12\\
        & Proposed & 6.00& 4.67 & 0.48\\
        \hline
        \multirow{2}{*}{2)} & \cite{Jang2024RLWE} & 21.00 & 11.52 &1.12 \\
        & Proposed & 8.00& 6.46 & 0.52\\
        \hline
    \end{tabular}
\end{table}

For both cases, the scale factors are chosen as $(\scaL, \scas_1,\scas_2)=(10^{-10},10^{-4},1)$.
The proposed encrypted controller \eqref{eq:ctrl_enc} and that of \cite{Jang2024RLWE} are 
implemented using Lattigo version $6.1.0$ \cite{lattigo}, an HE library written in Go.
The encryption parameters are set as $q\approx 2^{56}$, $P\approx 2^{51}$, and $(N, \nu, \sigma) = (2^{13},q,19.2)$, ensuring $128$-bit security \cite{HEstandard}.
All experiments are performed on a desktop with Intel Core i7-12700K CPU at 3.61 GHz,
and compiled with Go 1.23.5.
The entire code for simulation is available at GitHub\footnote{\url{https://github.com/CDSL-EncryptedControl/CDSL}}.

\begin{figure}
    \vspace{-6pt}
    \input{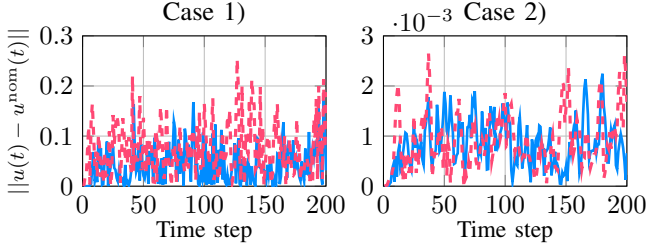}
    \centering
    \caption{Performance error of the encrypted controller of \cite{Jang2024RLWE} (red dashed-line) and the proposed encrypted controller (blue solid line).} 
    \label{fig:performance_error}
\end{figure}

Table~\ref{tab:simul_time} shows the elapsed time for each control period---from packing of the plant output to unpacking of the control input---of encrypted controllers.
It is evident from Table~\ref{tab:simul_time} that
the proposed design remarkably alleviates the computational burden of encrypted controllers, compared to that of \cite{Jang2024RLWE}.
In particular, the elapsed time of \cite{Jang2024RLWE} strictly increases as the degree of the controller increases from Case 2) to Case 1), while that of the proposed design decreases since Case 1) has a smaller $\kappa$.

Fig.~\ref{fig:performance_error} depicts the performance error $\|u(t)-u^\mathrm{nom}(t) \|$ of the encrypted controllers,
where $u(t)$ is the control input generated by the encrypted controller and $u^\mathrm{nom}(t)$ is the one from the original (unencrypted) controller.
It demonstrates that the performance error remains bounded for both cases.

\appendix

\addtolength{\topmargin}{5pt}
\begin{algorithm}[t]
\caption{Finding a similarity transformation for the rational canonical form}
\begin{algorithmic}[1]
\renewcommand{\algorithmicrequire}{\textbf{Procedure}}
\renewcommand{\algorithmicensure}{\textbf{Input}}
\Ensure $F\in\Q^{n\times n}$
\Require 
\State Let $l$, $p_i(s)$ and $\eta_i$ for $i=1,\,\ldots,\,l$ be defined by \eqref{eq:minpoly}.
\For{($i = 1; i = l; i=i+1$)}
\State $d_i \gets \deg p_i(s)$, $k \gets 0$, $\mathcal{B}_i \gets \emptyset$
    \For{($j = \eta_i; j = 1; j = j-1$)}
         \State $r_j \gets \left(\nult(p_i(F)^j) - \nult(p_i(F)^{j-1})\right)/d_i$
        \While{($k < r_j$)}
        \State Find $w \in \ker p_i(F)^j \setminus\ker p_i(F)^{j-1}$ such that $ \mathcal{B}_i \cup \{p_i(F)^{j-1}w \}$ is linearly independent.
        \State $\mathcal{B}_i \gets \mathcal{B}_i\cup \{w, Fw, \ldots, F^{jd_i-1}w\}$, $k \gets k+1$
        \State $v_{i,k} \gets w$, $\delta_{i,k} \gets j$
        \EndWhile
    \EndFor
    \State $\kappa_i \gets k$
\EndFor
\State $\kappa \gets \max_i \kappa_i$
\For{($j=1, j=\kappa, j = j+1$)}
    \State $\mathcal{I}_j\gets \lbrace 1\leq i\leq l: v_{i,j}\,\,\text{exists}\rbrace$
    \State
    $v_j \gets \sum_{i\in\mathcal{I}_j}v_{i,j}$, $\delta_j \gets \sum_{i\in\mathcal{I}_j}d_i\delta_{i,j}$
    \vspace{2pt}
    \State $V_j \gets \begin{bmatrix}
        F^{\delta_j-1}v_j & \cdots & Fv_j & v_j
    \end{bmatrix}$
\EndFor
\State $V \gets \begin{bmatrix}
    V_\kappa & \cdots & V_2 & V_1
\end{bmatrix}$
\State \Return $T:= V^{-1}$
\end{algorithmic}
\label{alg:RCF}
\end{algorithm}

We present a constructive method to transform a square matrix into its rational canonical form. Given $F \in \Q^{n\times n}$, Algorithm~\ref{alg:RCF} computes a similarity transformation $T \in \Q^{n\times n}$ so that $TFT^{-1}$ is the rational canonical form of $F$.

To understand Algorithm~\ref{alg:RCF}, we review some linear algebra \cite{hoffman1971algebra}. 
For a nonzero $v \in \Q^{n}$, there exists a unique $g\in \N$ such that the set $\{v,\, Fv,\, \ldots,\, F^{g-1}v\}$ is linearly independent but the set $\{v,\, Fv,\, \ldots,\, F^{g}v\}$ is linearly dependent. 
Then, it can be observed that $V_v^\dagger FV_v$, where $V_v:=[F^{g-1}v,\,\ldots,\,Fv,\,v]$ and $V_v^\dagger$ is its left inverse, is a companion matrix. 
The set $\{v,\, Fv,\, \ldots,\, F^{g-1}v\}$ is called \emph{the $F$-cyclic basis generated by} $v$. 
Thus, we construct the set of the columns of $V = T^{-1}$, as a union of $F$-cyclic bases.

Let the minimal polynomial $\mu_F(s)$ of $F$ be written by
\begin{equation}\label{eq:minpoly}
    \mu_F(s) = \textstyle\prod_{i=1}^l p_i(s)^{\eta_i},
\end{equation}
with some $\{ \eta_i \}_{i=1}^{ l }$, where the polynomial $p_i(s)$ is monic and irreducible over $\Q$, for $i=1,\,\ldots,\,l$.
The primary decomposition theorem \cite[Theorem 6.12]{hoffman1971algebra} ensures that $\Q^n = \oplus_{i=1}^{l}\ker p_i(F)^{\eta_i}$, where $\oplus$ denotes the direct sum. 
Thus, in the lines 2--13 of Algorithm~\ref{alg:RCF}, we first find a basis of each $\ker p_i(F)^{\eta_i}$ as a union of $F$-cyclic bases.
Subsequently, in the lines 14--19, the generators of these $F$-cyclic bases are added to each other to build the desired generators.

\begin{proposition}\label{prop:kappa}
    Let $T$ and $\kappa$ be determined from Algorithm~\ref{alg:RCF}. Then, $TFT^{-1}$ is the rational canonical form of $F$. 
    Furthermore, $\kappa =  \max_{\lambda \in \mathbb{C}} \mathrm{nullity}(F-\lambda I)$.
\end{proposition}

\begin{proof}[Sketch of proof] 
    Consider the lines 2--13. We claim that the resulting $\mathcal{B}_i$ after the line 13 is a basis for each $\ker p_i(F)^{\eta_i}$, for $i = 1, 2, \ldots, l$. Fix $i$ and consider a chain of spaces:
    \begin{equation} \label{eq:chainKernel}
        \{ 0  \} \subseteq \ker p_i(F)  {\subseteq} \cdots  {\subseteq}  \ker p_i(F)^{\eta_i - 1} {\subseteq} \ker p_i(F)^{\eta_i}.
    \end{equation}
    Our strategy is to find generators in $\ker p_i(F)^j \setminus \ker p_{i}(F)^{j-1}$ in the descending order of $j = \eta_i, \eta_i-1, \ldots, 1$.
    Since $d_i | \nult(p_i(F)^j)$ \cite[Theorem 7.24]{FIS2002Linear}, $r_j$ as defined in the line 5 is a positive integer for $j = \eta_i, \eta_i-1, \ldots, 1$. 
    We also note that $r_{j+1} \le r_{j}$ for all $j$. 
    Then, the while-loop in the lines 6--10 finds ($r_{j} - r_{j+1}$)-generators, where the variable $k$ denotes the number of generators that are already selected.

    We inductively show that $\mathcal{B}_i$ is linearly independent. At $k = 0$ and $j = \eta_i$, the existence of $w \in \ker p_i(F)^j \setminus\ker p_i(F)^{j-1}$ is clear, because of the minimality of $\mu_F(s)$.
    In addition, it can be easily shown that the $F$-cyclic basis generated by $w$ is $\{ w, Fw, \ldots, F^{jd_i-1}w\}$. 
    Now, suppose that $\mathcal{B}_i$ consists of $F$-cyclic bases generated by $v_{i,\iota} \in \ker p_i(F)^{\delta_{i,\iota}}\setminus \ker p_i(F)^{\delta_{i,\iota} -1}$ for $\iota = 1, \ldots, k$, and $\mathcal{B}_i =\cup_{\iota=1}^k\{v_{i,\iota}, Fv_{i,\iota}, \ldots, F^{d_i\delta_{i,\iota}-1}v_{i,\iota}\}$ is linearly independent. 
    The dimension condition $k < r_j$ ensures the existence of $w \in \ker p_i(F)^j \setminus\ker p_i(F)^{j-1}$ such that $\mathcal{B}_i \cup \{p_i(F)^{j-1}w \}$ is linearly independent, as in the line 7. 
    Then we show that $\mathcal{B}_i\cup \{w, Fw, \ldots, F^{jd_i-1}w\}$ is linearly independent. 
    Suppose that
    $
    \textstyle\sum_{\iota=1}^k\gamma_\iota(F)v_{i, \iota} + \gamma_{k+1}(F)w = \mathbf{0}
    $
    for some $\gamma_\iota(s)\in \Q[s]$ of degree less than $d_i\delta_{i,\iota}$, for $\iota = 1, \ldots, k$, and $\gamma_{k+1}(s) \in \Q[s]$ of degree less than $d_ij$. 
    Indeed, we claim that $\gamma_\iota(s) = 0$ for all $\iota$. The division algorithm provides that
    \begin{multline} \label{eq:noname}
        \textstyle\sum_{\iota=1}^k\sum_{h=0}^{\delta_{i,\iota}-1}p_i(F)^{h} q_{h,\iota}(F)v_{i,\iota}\\+ \textstyle\sum_{h=0}^{j-1}p_i(F)^hq_{h,k+1}(F)w = \mathbf{0},
    \end{multline}
    where $q_{h,\iota}(s) = 0$ or $\deg(q_{h,\iota}(s)) < d_i$ for all $h$ and $\iota$. By multiplying $p_i(F)^{j}$ on the both sides of \eqref{eq:noname}, it gives $q_{h,\iota}(s) = 0$ for $h = 0,\ldots, \delta_{i, \iota}-j-1$ and for $\iota = 1, \ldots, k$, thanks to the linear independence of $\mathcal{B}_i$. Also, one can show that multiplying $p_i(F)^{\lambda}$ on the both sides of \eqref{eq:noname} inductively for $\lambda = j-1, \ldots, 1$ gives $q_{h,\iota}(s) = 0$ for all $h$ and $\iota$, with the irreducibility of $p_i(s)$ and Bézout identity. 
    Therefore, we obtain $\gamma_\iota(s) = 0$ for all $\iota$.
    Finally, since $\nult (p_i(F)^{\eta_i}) = d_i \sum_{j=1}^{\eta_i} r_j = d_i\sum_{\iota=1}^{\kappa_i}\delta_{i, \iota}$,
    $\mathcal{B}_i$ is a basis for $\ker p_i(F)^{\eta_i}$. 

    Next, we consider the lines 14--19. 
    Let $\zeta_j(s)$ be the monic polynomial of minimal degree such that $\zeta_j(F)v_j = \mathbf{0}$.
    Then, it can be shown that $\zeta_j(s) = \prod_{i\in \mathcal{I}_j} p_{i}(s)^{\delta_{i,j}}$, that is, the $F$-cyclic basis generated by $v_j$ is exactly $\beta_j = \{v_j, \,Fv_j, \, \ldots, \, F^{\delta_j-1}v_j\}$. Also, it is clear that $\cup_{j=1}^{\kappa}\beta_j$ is linearly independent, thus a basis for $\Q^n$, as $\sum_{j=1}^{\kappa}\delta_j = n$.
    Moreover, note that $\zeta_j(s)$ is indeed the characteristic polynomial of the companion matrix $C_j = V_j^\dagger F V_j$ and $\Span(\beta_j)$ is $F$-invariant.
    Since $\mathcal{I}_j \subset \mathcal{I}_{j-1}$ for $j = 2, 3, \ldots, \kappa$, we have $\zeta_j(s)|\zeta_{j-1}(s)$ for $j = 2, 3, \ldots, \kappa$, which implies that $TFT^{-1}$ is the rational canonical form.
    
    Finally, recall that $\kappa = \max_{i} (\nult p_i(F))/d_i$. Furthermore, $p_i(s)$ and $p_j(s)$ have no common complex roots whenever $i \ne j$ \cite[Theorem 13.9]{dummit2004abstract}. 
    Therefore, we obtain $\kappa =  \max_{\lambda \in \mathbb{C}} \mathrm{nullity}(F-\lambda I)$.
\end{proof}

\section*{Acknowledgment}
The authors are grateful to Prof. Yongsoo Song
of Seoul National University
for the helpful discussions.

\bibliography{F_citation}
\bibliographystyle{IEEEtran}

\end{document}